\documentclass[11pt,a4paper]{article}
\usepackage{lmodern} 
\usepackage[T1]{fontenc}%

\usepackage{mathtools}
\usepackage{amsmath,amsfonts,amsthm,amssymb,color}
\usepackage[dvipsnames]{xcolor}
\usepackage{fancybox}
\usepackage[ocgcolorlinks]{hyperref}
\usepackage[ruled,vlined, linesnumbered]{algorithm2e}
\usepackage{tikz}
\usepackage{pifont}
\usepackage{url}
\usepackage{thm-restate}
\usepackage{tcolorbox}
\allowdisplaybreaks[1]
\usepackage{nicefrac}
\usepackage[labelfont=bf]{caption}
\usepackage[noabbrev,capitalize,nameinlink]{cleveref}

\usepackage[nottoc,numbib]{tocbibind}
\usepackage[backend=biber, isbn=false, style=alphabetic, backref=true, doi=false, url=false, maxcitenames=7, mincitenames=5, maxalphanames=7, maxbibnames=7, minbibnames=5, minalphanames=5, defernumbers=true, sortlocale=en_US]{biblatex}

\usepackage{fullpage}
\colorlet{DarkRed}{red!50!black}
\colorlet{DarkGreen}{green!50!black}
\colorlet{DarkBlue}{blue!50!black}

\hypersetup{
	linkcolor = DarkRed,
	citecolor = DarkGreen,
	urlcolor = DarkBlue,
	bookmarksnumbered = true,
	linktocpage = true
}

\newtheorem{theorem}{Theorem}[section]

\newtheorem{lemma}[theorem]{Lemma}
\newtheorem{invariant}[theorem]{Invariant}

\newtheorem{claim}[theorem]{Claim}

\newtheorem{assumption}[theorem]{Assumption}

\theoremstyle{definition}

\newtheorem{remark}[theorem]{Remark}

\newcommand\congestion{\mathit{cong}}

\DeclareFontFamily{U}{mathb}{\hyphenchar\font45}
\DeclareFontShape{U}{mathb}{m}{n}{<5> <6> <7> <8> <9> <10> gen * mathb
<10.95> mathb10 <12> <14.4> <17.28> <20.74> <24.88> mathb12}{}
\DeclareSymbolFont{mathb}{U}{mathb}{m}{n}
\DeclareMathSymbol{\rcirclearrow}{\mathbin}{mathb}{'367}

\def\wh{\widehat}

\def\wt{\widetilde}

\DeclareMathOperator{\Cong}{\mathtt{cong}}

\DeclareMathOperator{\dist}{\mathtt{dist}}

\renewcommand\O{\wt{O}}

\foreach \x in {a,...,z}{%
	\expandafter\xdef\csname b\x\endcsname{\noexpand\boldsymbol{\x}}
}

\foreach \x in {A,...,Z}{%
	\expandafter\xdef\csname m\x\endcsname{\noexpand\mathbf{\x}}
}

\foreach \x in {A,...,Z}{%
	\expandafter\xdef\csname c\x\endcsname{\noexpand\mathcal{\x}}
}

\newcommand{\Abs}[1]{\left\vert#1\right\vert}

\newcommand{\vol}{{{\textsf{vol}}}}

\newcommand\bwInt{\boldsymbol{{\mathit{w}}}}
\renewcommand\bw{\boldsymbol{\mathit{w}}}

\newcommand{\APSP}{\alpha_{\texttt{APSP}}}

\newcommand{\Set}[2]{\left\{#1 ~\middle\vert~ #2\right\}}

\usepackage{cancel}

\begin{document}

\title{A Simple Framework for Finding Balanced Sparse Cuts via APSP}
\author{
Li Chen\thanks{Li Chen was supported by NSF Grant CCF-2106444.}\\ Georgia Tech\\ lichen@gatech.edu
\and
Rasmus Kyng\thanks{The research leading to these results has received funding from the grant ``Algorithms and complexity for high-accuracy flows and convex optimization'' (no. 200021 204787) of the Swiss National Science Foundation.}\\ ETH Zurich \\ kyng@inf.ethz.ch 
\and
Maximilian Probst Gutenberg\footnotemark[2]\\ ETH Zurich\\ maxprobst@ethz.ch
\and
Sushant Sachdeva\thanks{Sushant Sachdeva's research is supported by an NSERC (Natural Sciences and Engineering Research Council of Canada) Discovery Grant.
} \\ University of Toronto \\ sachdeva@cs.toronto.edu
}
\date{}

\maketitle

\begin{abstract}
We present a very simple and intuitive algorithm to find balanced sparse cuts in a graph via shortest-paths. 
Our algorithm combines a new multiplicative-weights framework for solving unit-weight multi-commodity flows with standard ball growing arguments. 
Using  Dijkstra's algorithm for computing the shortest paths afresh every time gives a very simple algorithm that runs in time $\O(m^2/\phi)$ and finds an $\O(\phi)$-sparse balanced cut, when the given graph has a $\phi$-sparse balanced cut.
Combining our algorithm with known deterministic data-structures for answering approximate All Pairs Shortest Paths (APSP) queries under increasing edge weights (decremental setting), we obtain a simple deterministic algorithm that finds $m^{o(1)}\phi$-sparse balanced cuts in $m^{1+o(1)}/\phi$ time.
Our deterministic almost-linear time algorithm matches the state-of-the-art in randomized and deterministic settings up to subpolynomial factors, while being significantly simpler to understand and analyze, especially compared to the only almost-linear time deterministic algorithm, a recent breakthrough  by Chuzhoy-Gao-Li-Nanongkai-Peng-Saranurak (FOCS 2020).

\end{abstract}

\section{Introduction}
Graph partitioning is a fundamental algorithmic primitive that has been studied extensively. 
There are several ways to  formalize the question. We focus on the question of finding balanced separators in a graph. 
More precisely, given an $m$-edge graph $G=(V,E)$, the conductance of a cut is defined by $\Phi_G(S) = \frac{|E_G(S, \overline{S})|}{\min\{\vol(S),\vol{({V\setminus S})}\}}$  where $E_G(S, \overline{S})$ is the set of edges with exactly one endpoint in $S,$ and the volume of $S,$ denoted $\vol(S)$ is the sum of the degrees of vertices in $S.$
We say that a cut $(S, V\setminus S)$ is $b$-balanced if $\vol(S),\vol(V \setminus S) \ge b\cdot \vol(V).$
The objective in the Balanced Separator problem is 
\begin{quote}
    Given parameters $b, \phi \le 1,$ either find a cut $(S, V \setminus S)$ that is $b$-balanced and has conductance $\Phi_G(S) \le \phi,$ or certify that every $\Omega(b)$-balanced\footnote{Note that we allow the algorithm to return an $\Omega(b)$-balanced sparse cut when the graph has a $b$-balanced sparse cut. Such an algorithm is known as a pseudo-approximation algorithm. All known efficient algorithms for balanced cut find pseudo-approximations.} cut has conductance at least $\alpha\phi.$ 
\end{quote}

The Balanced Separator problem is a classic NP-hard problem and under the Small-Set-Expansion hypothesis, even NP-hard to approximate to within an arbitrary constant~\cite{RaghavendraST12}. Thus, the above formulation allows for $\alpha$-approximation for some $\alpha < 1.$
This problem has been studied extensively due to its application to divide-and-conquer on graphs, and theoretical connections to random walks, spectral graph theory, and metric embeddings.
\paragraph{Our Results.}
In this paper, we present a very simple and intuitive algorithm for Balanced Separator.
Our algorithm gives a simple framework based on (scalar) multiplicative weights that reduces the problem to computing approximate shortest paths in a graph under increasing lengths for the edges (decremental setting). Our framework either finds a balanced cut with small conductance, or certifies that every balanced cut has large conductance (\Cref{thm:fineGrainedReduction}).

If one simply uses Dijkstra's algorithm to compute the necessary shortest paths afresh each time, our algorithm gives an $\O(m^2/\phi)$ time algorithm that achieves approximation  $\alpha = \Omega(1/\log^2 n)$ for cuts of constant balance, and $\alpha = \Omega(1/\log n \cdot \log\log n)$ for cuts of constant balance and conductance (\Cref{theorem:balanceCutCorNaive}).
If we instead use known  $n^{o(1)}$-approximate deterministic dynamic algorithms for decremental All-Pairs-Shortest-Paths (APSP), we obtain an algorithm that runs in $m^{1+o(1)}/\phi$ and achieves an approximation of $\alpha = n^{o(1)}$ (\Cref{theorem:balanceCutCorAPSP}).

Our algorithm can be described very simply.
We attempt to embed an explicit expander $H$ as a multi-commodity flow using  paths of length $\O(\phi^{-1})$ in $G,$ while ensuring that the congestion on the edges in $G$ is at most $\O(\phi^{-1}).$ 
If the ends of points of an edge $e \in H$ are connected in $G$ using a short path, we use the  path in $G$ to route $e.$
Further, we increase the length of each edge on this path by a multiplicative factor. This increased length makes it less likely that this path will be used in the future. A simple multiplicative-weights argument here now allows us to bound the congestion over the course of entire algorithm.
If our algorithm succeeds in embedding most edges of $H$ in $G,$ this provides us a certificate that all balanced cuts in $G$  have expansion $\widetilde{\Omega}(\phi).$
If our algorithm fails, we find several edges of $H$ such that the ends points of these edges are at distance $\widetilde{\Omega}(\phi^{-1})$ as measured by the lengths of the edges computed by the algorithm. Now, we can apply a simple ball-growing argument to recover a balanced cut of conductance $\phi.$

\paragraph{Applications.} While finding the Balanced Sparsest Cut is a crucial ingredient in Divide-And-Conquer frameworks for many algorithms (see \cite{shmoys1997cut} for an introduction), and has various applications ranging from VLSI Design, Image Segmentation \cite{shi2000normalized} to PRAM emulation, we want to point out in particular that our algorithm can be used to replace the use of the Cut-Matching framework~\cite{KhandekarRV09} in the work of Saranurak-Wang~\cite{SaranurakW18} (see \Cref{rem:extractExpanderSW19} in \Cref{sec:seporcert}). Together, this gives an elegant framework for computing expander decompositions which in turn have been pivotal in various recent breakthroughs in algorithmic graph theory with applications to computing Electric Flows \cite{SpielmanT04}, Maximum Flows and Min-Cost Flows \cite{chen2022maximum}, Gomory-Hu Trees \cite{abboud2021subcubic, abboud2022breaking} for finding Global Min-Cuts deterministically \cite{kawarabayashi2018deterministic, li2020deterministic, li2021deterministic}, and many, many more.

\paragraph{Comparison to Previous Works.}
There has been a lot of work on algorithms for Balanced Separator.
The celebrated work of Leighton and Rao \cite{LeightonRao} showed that one could achieve an $O(\log n)$ approximation to Balanced Separator by repeatedly solving a linear program that computes a fractional multi-commodity flow.
Several works give a faster implementation of this approach via a multiplicative-weights algorithms for multi-commodity flow~\cite{PlotkinST95, Young95, GargK07, Fleischer00}, and by using the Leighton-Rao result as a black-box to deduce that they compute an $O(\log n)$ approximation. %
However, the running time they achieved for Balanced Separator was $\Omega(nm^2)$ since they repeatedly find and remove low-conductance cuts, each of which might be highly unbalanced, possibly introducing a factor of $n.$
In contrast, our algorithm works directly with balanced cuts, rather than multi-commodity flows. Our algorithm is in the same spirit as the Garg-K\"{o}nemann, Fleischer framework from \cite{GargK07, Fleischer00}, but directly incorporates the Leighton-Rao algorithm for finding low conductance cuts.

The groundbreaking work of Spielman and Teng on solving Laplacian linear systems~\cite{SpielmanT04} introduced the notion of local algorithms for finding low-conductance cuts, where the running time of the algorithm scales almost-linearly with the smaller size of the output cut. Thus the algorithm can be applied repeatedly to find balanced cuts in almost-linear time. 
Inspired by this work, multiple local algorithms were proposed~\cite{AndersenCL07, AndersenP09}. While all these algorithms are fast, and almost-linear in running time, they are inherently randomized, and the balanced cut found has conductance $\widetilde{\Omega}(\sqrt{\phi})$.%
In contrast, our algorithm is deterministic, and finds a cut of conductance at most $\phi\cdot m^{o(1)}.$

Another line of work develops fast SDP algorithms based on matrix-multiplicative weights. The most popular of these is the Cut-Matching framework of Khandekar-Rao-Vazirani~\cite{KhandekarRV09}. %
Inspired by~\cite{KhandekarRV09}, several works~\cite{AroraK07, OrecchiaSVV08, OrecchiaV11, OrecchiaSV12} obtained almost-linear time algorithms for Balanced Separator building on the matrix-multiplicative weights framework. 
While the cut-matching framework and the resulting algorithms are elegant, they rely on rather involved techniques that are non-intuitive. 
The celebrated work of Arora-Rao-Vazirani~\cite{AroraRV09} obtained an $O(\sqrt{\log n})$ approximation for Balanced Separator via an SDP based algorithm. Faster algorithms built on their ideas~\cite{AroraHK10, Sherman09} achieved almost-linear running time with $O(\sqrt{\log n})$ approximation.
However, these algorithms are very involved, based on matrix-multiplicative weights,  randomized, and rely on near-linear time (approximate) max-flow. 
Our algorithm and analysis work with scalar multiplicative weights and are very simple to understand. Further, our algorithm only need to invoke approximate shortest-path oracles under increasing edge weights.

The only previous deterministic, almost-linear time approximation algorithm for Balanced Separator was given recently by Chuzhoy-Gao-Li-Nanongkai-Peng-Saranurak~\cite{ChuzhoyGLNPS19}. Their algorithm relies on a rather intricate recursive scheme that implicitly uses at each recursion level a reduction to decremental APSP. But even the analysis on a single level relies on the rather involved expander pruning framework. In contrast to their work, the simplicity of our algorithm and analysis stands out.

We also point out that a generalization of \cite{ChuzhoyGLNPS19} to weighted graphs was given by Li and Saranurak \cite{li2021deterministic}. This algorithm implicitly uses \cite{ChuzhoyGLNPS19}, and is therefore even more involved.

\section{Main Result}
We formally state our results in this section.
Our main result is the following theorem.
\begin{theorem}\label{thm:conductanceVersion}
Given an $n$-vertex, $m$-edge graph $G$, an $\APSP$-approx decremental APSP algorithm and conductance parameter $\phi$ and balance parameter $b \in [1/n, 1/4]$, the algorithm $\textsc{LowConductanceCutOrCertify}(G, \phi, b)$ either
\begin{enumerate}
    \item\label{case:balanceLCCut} Returns a cut $(S, \overline{S})$ with $\vol_G(S), \vol_G(\overline{S}) \ge b \cdot \vol(G)$ with conductance $\Phi_G(S) \leq \phi$, or
    \item\label{case:noBalanceLCCut} Certifies that every cut $(X, \overline{X})$ with $\vol_G(X), \vol_G(\overline{X}) = \Omega( b \cdot \vol_G(G))$ has conductance at least  $\phi \cdot \Omega\left(\frac{1}{\APSP \log n \cdot \log(1/b) \cdot \log(\log(n)\APSP/(b\phi))}\right) = \phi \cdot \Omega\left(\frac{1}{\APSP \log^3(n)}\right).$
\end{enumerate}
The algorithm is deterministic and requires the APSP data structure to undergo $O( \APSP \cdot m \phi^{-1} \log^3 n)$ updates, queries it $O(m)$ times and spends an additional $O( \APSP \cdot m  \phi^{-1} \log^3 n)$ time.
\end{theorem}
\begin{remark}
APSP data structures often answer queries in time proportional to the number of edges on the approximate shortest path that they return. Our algorithm ensures that the number of such edges on all paths is bound by $O( \APSP \cdot m \phi^{-1} \log^3 n)$.
\end{remark}

We note that for computing balanced cuts (i.e. cuts where $b$ is constant) which is arguably the most interesting case, our approximation guarantee becomes $\Omega(1/ \APSP\log n \cdot \log (\APSP \phi^{-1} \cdot \log n))$. For a decremental APSP data structure with constant-approximation and $\phi \ge \Omega(1/\log^{O(1)} n)$, this further simplifies to $\Omega(1/\log n \log\log n)$.

Using the efficient $n^{o(1)}$-approximate decremental APSP data structure from \cite{BGS21} or \cite{Chuzhoy21}, we obtain the following result\footnote{We remark that both data structures \cite{BGS21, Chuzhoy21} implicitly rely on the framework of Chuzhoy-Gao-Li-Nanongkai-Peng-Saranurak~\cite{ChuzhoyGLNPS19}, thus, our reduction in combination with these data structures does not yield a simpler algorithm in itself. We are however optimistic that simpler data structures for the decremental APSP problem are available in the future that do not necessarily rely on expander techniques.}:
\begin{theorem}
\label{theorem:balanceCutCorAPSP}
Given an $n$-vertex, $m$-edge graph $G$, a conductance parameter $\phi$ and balance parameter $b \in [1/n, 1/4]$, there is an algorithm $\textsc{LowConductanceCutOrCertify}(G, \phi, b)$ that can either
\begin{enumerate}
    \item\label{casecor:balanceCut} Find a cut $(S, \overline{S})$ with $\vol_G(S), \vol_G(\overline{S}) \ge b \cdot \vol(G)$ with conductance $\Phi_G(S) \leq \phi$, or
    \item\label{casecor:noBalanceCut} Certify that every cut $(X, \overline{X})$ with $\vol_G(X), \vol_G(\overline{X}) = \Omega( b \cdot \vol_G(G))$ has conductance $\phi/n^{o(1)}.$
\end{enumerate}
The algorithm is deterministic and runs in $m^{1+o(1)} / \phi$ time.
\end{theorem}

On the other hand, one can run Dijkstra's shortest path algorithm for every query and obtain the following:
\begin{theorem}
\label{theorem:balanceCutCorNaive}
Given an $n$-vertex, $m$-edge graph $G$, a conductance parameter $\phi$ and balance parameter $b \in [1/n, 1/4]$, there is a deterministic algorithm $\textsc{LowConductanceCutOrCertify}(G, \phi, b)$ that can either
\begin{enumerate}
    \item\label{caseCor2:balanceCut} Find a cut $(S, \overline{S})$ with $\vol_G(S), \vol_G(\overline{S}) \ge b \cdot \vol(G)$ with conductance $\Phi_G(S) \leq \phi$, or
    \item\label{casecor2:noBalanceCut} Certify that every cut $(X, \overline{X})$ with $\vol_G(X), \vol_G(\overline{X}) = \Omega( b \cdot \vol_G(G))$ has conductance $\phi \cdot \Omega\left(\frac{1}{\log n \cdot \log(1/b) \cdot \log(\log(n)/(\phi b))}\right).$
\end{enumerate}
The algorithm is deterministic and runs in $\O(m^2 / \phi)$ time.
\end{theorem}

\section{Preliminaries}

\paragraph{Sparsity and Expanders.} In this article, we consider an undirected $n$-vertex graph $G = (V,E)$. For such a graph, we define the sparsity of a cut $\emptyset \subsetneq S \subsetneq V$ by $\Psi_G(S) = \frac{|E_G(S, \overline{S})|}{\min\{|S|, |\overline{S}|\}}$ where $E_G(S, \overline{S})$ is the set of edges with exactly one endpoint in $S$. The sparsity of a graph $G$ is defined $\Psi(G) = \min_{\emptyset \subsetneq S \subsetneq V} \Psi(S)$. If $G$ contains no $\psi$-sparse cut, we say that $G$ is a $\psi$-expander.

\paragraph{Conductance vs. Sparsity.} Via a simple reduction replacing each vertex of degree $d$ with an explicit expander graph on $d$ vertices (see~\cref{sec:condViaSpars}), we can reduce to the case where every vertex has degree at most 10. 
In such a graph, for any set $S \subseteq V,$ $|S| \le \vol(S) \le 10|S|,$ and thus, instead of conductance $\Phi_G(S) = \frac{|E_G(S, \overline{S})|}{\min\{\vol(S),\vol{({V\setminus S})}\}},$ we can work with sparsity $\Psi_G(S) = \frac{|E_G(S, \overline{S})|}{\min\{|S|,|V\setminus S|\}}.$ 
Throughout the rest of the article, we will therefore work with sparsity instead of conductance.

\paragraph{Expander Constructions.} Given any $n$, there is a deterministic construction of a $\Omega(1)$-expander on $n$ vertices of bounded degree.
This will be an essential tool used in our proof and we use $\psi_0$ to denote the universal lower bound on the sparsity of such family of expanders.
\begin{theorem}[See Thm. 2.4 of \cite{ChuzhoyGLNPS19} based on Thm 2 of \cite{GabberG81}.]\label{thm:constExp}
There is an universal constant $\psi_0 \in (0, 1)$ and an algorithm $\textsc{ConstDegExpander}(n)$ that returns a $\psi_0$-expander $H$ on a vertex set of size $n$ with maximum degree $9$. The algorithm runs in time $O(n)$.
\end{theorem}
\begin{remark}\label{rmk:simpleExpander}
While deterministic algorithms to construct a constant-degree, constant sparsity expander require rather involved proof techniques, we prove in \Cref{sec:randExpand} a simple randomized algorithm to construct a $O(\log n)$-degree $\Omega(\log n)$-expander $H$ in $O(n \log n)$ time. Using this randomized algorithm in place of the above theorem only affects guarantees of our overall algorithm by polylogarithmic factors. 
\end{remark}

\paragraph{Graph Embeddings.} Given graphs $H$ and $G$ that are defined over the same vertex set, then we say that a function $\Pi_{H \mapsto G}$ is an \emph{embedding} if it maps each edge $(u,v) \in H$ to a $u$-to-$v$ path $P_{u,v} = \Pi_{H \mapsto G}(u,v)$ in $G$. We say that the \emph{congestion} of $\Pi_{H \mapsto G}$ is the maximum number of times that any edge $e \in E(G)$ appears on any embedding path: \[
\congestion(\Pi_{H \mapsto G}) = \max_{e \in E(G)} |\{ e' \in E(H) \;|\; e \in \Pi_{H \mapsto G}(e') \}|.
\]

\paragraph{Certifying Expander Graphs via Embeddings.} Graph embeddings are useful since they allow us to argue that if we can embed a graph $H$ that is known to be an expander into a graph $G$, then we can reason about the sparsity of $G$, as shown below.

\begin{lemma}\label{lma:folklore_embeddingEasy}
Given a $\psi$-expander graph $H$ and an embedding of $H$ into $G$ with congestion $C$, then $G$ must be an $\Omega\left(\frac{\psi}{C}\right)$-expander. 
\end{lemma}
\begin{proof}
Consider any cut $(S, V \setminus S)$ with $|S| \leq |V \setminus S|$. Since $H$ is a $\psi$-expander, we have that $|E_H(S, V \setminus S)| \geq \psi|S|$. We also know by the embedding of $H$ into $G$, that for each edge $(u,v) \in E_H(S, V \setminus S)$, we can find path a $P_{u,v}$ in $G$ that also has to cross the cut $(S, V \setminus S)$ at least once. But since each edge in $G$ is on at most $C$ such paths, we can conclude that at least $|E_H(S, V \setminus S)|/ C \geq \psi|S|/C$ edges in $G$ cross the cut $(S, V \setminus S)$.
\end{proof}

We use the following generalization of this Folklore result to balanced sparse cuts.

\begin{restatable}{lemma}{folkloreEmbed}\label{lma:folklore_embedding}
Given a $\psi$-expander graph $H$, a subgraph $H' \subseteq H$ with $|E(H \setminus H')| \leq \frac{\psi}{2}bn$ for some $b \in [0,1]$ and an embedding $\Pi_{H' \mapsto G}$ of $H'$ into $G$ with congestion $C$, then for all cuts $(S, \overline{S})$ where $bn \leq |S| \leq n/2$, we have $\Psi_G(S) = \Omega\left(\frac{\psi}{C}\right)$. 
\end{restatable}
\begin{proof}
Observe that for each such $(S, \overline{S})$, we have $|E_{H'}(S, \overline{S})| \geq |E_H(S, \overline{S})| - |E(H \setminus H')| \geq \psi|S| - \frac{\psi}{2}bn \geq \frac{\psi}{2}|S|$.
Using the same argument as above, the cut size of $S$ in $G$ is at least $\Abs{E_{G}(S, \overline{S})} \ge \Abs{E_{H'}(S, \overline{S})} / C \ge \psi |S| / 2C.$
\end{proof}

\paragraph{Decremental All-Pairs Shortest-Paths (APSP).} 
A decremental $\APSP$-approximate All-Pairs Shortest-Paths (APSP) data structure (abbreviated $\APSP$-APSP) is a data structure that is initialized to an $m$-edge $n$-vertex graph $G$ and supports the following operations:
\begin{itemize}
    \item $\textsc{IncreaseEdgeWeight}(u,v, \Delta)$: increases the edge weight of $(u,v)$ by $\Delta$. 
    \item $\textsc{QueryDistance}(u,v)$: for any $u,v \in V$ returns a distance estimate $\tilde{d}(u,v)$ that $\APSP$-approximates the distance from $u$ to $v$ in the current graph $G$ denoted $d_G(u,v)$, i.e. $\tilde{d}(u,v) \in [d_G(u,v), \APSP \cdot d_G(u,v)]$.
    \item $\textsc{QueryPath}(u,v)$: returns a path $\pi$ from $u$ to $v$ in the current graph $G$ of total weight $\tilde{d}(u,v)$ (that is the value of the distance estimate if queried).
\end{itemize}
We denote the total time required by the data structure to execute a series of $q$ queries and $u$ update operations on an $n$-vertex constant-degree graph by $T_{APSP}(q,u)$.

Recently, deterministic $n^{o(1)}$-approximate APSP data structures have been developed (see \cite{Chuzhoy21,BGS21}) that process any sequence of $\tilde{O}(m)$ edge weight increases in total time $m^{1+o(1)}$ while answering distance queries in time $n^{o(1)}$ time and for a path query, returns paths in time near-linear in the number of edges on the path (i.e. if it returns a path $P$, it takes at most time $|P|n^{o(1)}$. We conjecture that in the near-future, $O(\log n)$-APSP data structures are found that implement edge weight increases in time $\tilde{O}(m)$ and answers distance queries in time $\tilde{O}(1)$ and path queries in time $\tilde{O}(|P|)$. 

\section{Our Algorithm}

In this section, we present an algorithm to find sparse cuts with respect to sparsity or embed an expander into a constant-degree graph $G$. By standard reductions (given in \Cref{sec:condViaSpars} and \Cref{sec:constDegAssump}), one can translate between sparsity and conductance and remove the bounded-degree assumption, both with only a constant loss in quality. Thus, by proving the theorem below, we directly establish our main result, \Cref{thm:conductanceVersion}.

\begin{restatable}{theorem}{APSPreduction}\label{thm:fineGrainedReduction}
\label{theorem:balanceCut}
Given a graph $G$ of degree at most 10, an $\APSP$-approx decremental APSP algorithm and sparsity parameter $\psi$ and balance parameter $b \in [1/n, 1/4]$, there is an algorithm \\ $\textsc{SparseCutOrCertify}(G, \psi, b)$ (\Cref{alg:mainAlgo}) that can either
\begin{enumerate}
    \item\label{case:balanceCut} Find a cut $(S, \overline{S})$ with $|S|, |\overline{S}| \ge bn$ of sparsity $\leq \psi$, or
    \item\label{case:noBalanceCut} Certify that every cut $(X, \overline{X})$ with $|X|, |\overline{X}| = \Omega(bn)$ has sparsity
    \\ $\psi \cdot \Omega\left(\frac{1}{\APSP \log n \cdot \log(1/b) \cdot \log(\log(n)\APSP/(b\psi))}\right).$
\end{enumerate}
The algorithm is deterministic and requires the APSP data structure to undergo $O( \APSP \cdot n / \psi \log^3 n)$ updates, queries it $O(n)$ times and spends an additional $O( \APSP \cdot n / \psi \log^3 n)$ time.
\end{restatable}
\begin{remark}
Our algorithm ensures that the total number of edges summed across all queried paths is bound by $O( \APSP \cdot n / \psi \log^3 n)$.
\end{remark}

The algorithm contains two phases.
The first phase tries to embed an $\Omega(1)$-expander into the input graph $G$ with congestion $\O(1/\psi).$
Let $F$ be the subset of expander-edges the algorithm cannot embed.
If $|F| = O(b n)$, i.e. the algorithm embed all but $O(b n)$ edges, \Cref{lma:folklore_embedding} ensures that every $b$-balanced cut has sparsity $\wt{\Omega}(\psi).$
Otherwise, $|F| = \Omega(b n)$ and the algorithm outputs an edge weight $\bw$ such that every $(u, v) \in F$ are far apart w.r.t. $\bw.$
In this case, the second phase is initiated to extract a sparse $\Omega(b)$-balanced cut from these far-apart pairs of vertices.

\subsection{An Algorithm to Separate Or Certify}
\label{sec:seporcert}
First, we present the algorithm for the first phase that either embeds a large portion of an expander or finds a large set of far-apart vertex-pairs w.r.t. some edge weights $\bw.$

\begin{lemma}
\label{lemma:SeparateOrCertifyBalanced}
Given an $\APSP$-APSP data structure, two graphs $G$ and $H$ over the same vertex set $V$, a congestion parameter $C \in [1, n]$, and a balance parameter $b \in [1/n, 1/2]$. The algorithm $\textsc{SeparateOrCertify}(G, H, C, b)$ (\Cref{alg:embedOrSep}) outputs either
\begin{enumerate}
    \item A set of weights $\bwInt \in \mathbb{R}_{\geq 1}^{E(G)}$ with $\|\bwInt\|_1 \leq 20n$, a number $b' \in [b, 1/2]$, and a subset of edges $F \subseteq E(H)$ with $|F| > 10 b' n$ such that 
    \[ \forall (u,v) \in F, \quad \dist_{\bwInt}(u,v) > \frac{C}{b'}, \text{ or }\]
    
    \item A graph $H' \subseteq H$ with $|E(H) \setminus E(H')| \leq 10bn$ and an embedding $\Pi_{H' \mapsto G}$ that maps each edge $(u,v)$ in $H'$ to a $uv$-path in $G$ with congestion $O(C \cdot \APSP \cdot \log(1/b) \cdot \log(C \cdot \APSP / b))$.
\end{enumerate}
The algorithm is deterministic and requires the APSP data structure to undergo $O(C \APSP n \log^2 n)$ edge updates and $O(n)$ distance queries along with additional $O(C \APSP n \log^2 n)$ time.
    
\end{lemma}

\begin{algorithm}
\DontPrintSemicolon
$H' = (V, \emptyset)$; $\Pi_{H' \mapsto G} \gets \emptyset$; $\bw \gets \mathbf{1}^{|E(G)|}$; $\eta \gets  \frac{1}{ 4 C \APSP \log_2(10/b)}$.\label{lne:init}\\
Maintain an $\APSP$-approximate APSP data structure on $G$ weighted by $\bw$.\\
\For(\label{lne:forLoop}){$i = 0, 1, \ldots, \lfloor \log_2(1 / b) \rfloor$}{
    \ForEach(\label{lne:foreachUnrouted}){$e = (u,v) \in E(H) \setminus E(H')$}{
        \If(\label{lne:ifReallyEmbed}){$\textsc{APSP}.\textsc{QueryDist}(u,v) \leq 2^{i} \cdot C \APSP$}{
            Add $e$ to $H'$; $\Pi_{H' \mapsto G}(e) \gets \textsc{APSP}.\textsc{QueryPath}(u,v)$. \\
            \ForEach{$f \in \Pi_{H' \mapsto G}(e)$}{
             $\textsc{APSP}.\textsc{IncreaseWeight}(e, \eta \bw_e)$;\label{lne:increaseWeights}
             $\bw_e \gets (1+\eta) \bw_e$.
            }
        }
    }
    \lIf{$|E(H) \setminus E(H')| > 10 n / 2^i$}{
        \Return $(\bw, 2^{-i}, E(H) \setminus E(H'))$. \label{lne:terminateCase2} \label{lne:bPrimeDefn}
    }
}
\Return $(H', \Pi_{H' \mapsto G})$. \label{lne:terminateEmbed}
\caption{$\textsc{SeparateOrCertify}(G, H, C, b)$}\label{alg:embedOrSep}
\end{algorithm}

\paragraph{The Algorithm.} \Cref{alg:embedOrSep} implements $\textsc{SeparateOrCertify}(G, H, C, b)$. Here, the task of finding an embedding of $H$ into $G$ is interpreted as a multicommodity flow problem, that is each edge $(u,v) \in H$ gives rise to the demand to route one unit of flow from $u$ to $v$. Later, we use a $\psi_0$-expander in place of $H$. 

The goal of the algorithm is to find such an embedding/ multicommodity flow with small congestion which combined with our choice of $H$ certifies that $G$ is a good (almost) expander (i.e. contains no balanced sparse cut). Here, we guess the congestion to be roughly $C$ and want to enforce $\Cong(\Pi_{H \mapsto G}) \leq C$. In fact, we even provide a slightly tighter analysis.

To achieve this goal, we use a technique which is an instance of the \emph{Multiplicative Weight Update (MWU)} framework. Initially, we define a uniform weight function $\bw$ with weights over $G$. We try to embed each edge $(u,v) \in E(H)$ using a short $uv$-path $P_{uv}$ in $G$ with respect to $\bw$. Whenever we embed an edge $(u,v)$ in such a way and the path $P_{uv}$ contains an edge $e \in E(G)$, we increase the weight $\bw_e$ by a multiplicative factor $(1+\eta)$. Naturally, after $t$ edges have been embedded by using the edge $e$, we have scaled up the weight of $e$ by a factor of $(1+\eta)^t$. Using $e^x \le (1 + 2x), x \in [0, 1]$, and setting $\eta \approx C$ ensures that the weight $\bw_e$ approaches a large polynomial in $n$ for $t \gg 2\eta \log n$ (which again is $\approx C$). 

At the same time, the algorithm only embeds edges $(u,v) \in E(H)$ if the distance between the endpoints in $G$ w.r.t. $\bw$ is small. This ensures that $\|\bw\|_1 = O(n \log(1/b))$ and that we never use an edge $e$ into which many embedding paths are already routed. 

More precisely, we proceed in rounds to embed edges in $H$. At later rounds (i.e. when $i$ large), we have already embed a large number of edges in $H$. Since the number of remaining edges is small, we allow for them to be embed with slightly longer paths which still lets us argue that $\|\bw\|_1$ is increased by at most $O(n)$ in the current round. If in any round, it is not possible to embed many of the remaining edges with paths of weight at most the current threshold, we can simply return these edges and end up in the first scenario.

\paragraph{Correctness (Returning in \Cref{lne:terminateCase2}).} We start by proving the following claim which then immediately establishes correctness if \Cref{alg:embedOrSep} terminates at \Cref{lne:terminateCase2} (i.e. in the second scenario). 

\begin{invariant}\label{inv:totalWeight}
After the $i$-th iteration of the for-loop in \Cref{lne:forLoop}, we have $\|\bw\|_1 \leq 10n (1 + 2 \eta C \APSP \cdot (i+1)) \leq 20n$.
\end{invariant}
\begin{proof}
Initially, $\|\bw\|_1 = \|\mathbf{1}^{|E(G)|}\|_1 \leq 10n$.

To gauge the increase in $\|\bw\|_1$ during the $i$-th iteration of the for-loop, consider the effect of embedding a new edge $e$ in the foreach-loop starting in \Cref{lne:foreachUnrouted} (we only consider such iterations if the if-statement in \Cref{lne:ifReallyEmbed} evaluates true as otherwise $\bw$ does not change). 
Letting $\bw^{OLD}$ denote $\bw$ just before the foreach-loop iteration and $\bw^{NEW}$ right after.
We clearly have that $\|\bw^{NEW}\|_1 = \|\bw^{OLD}\| + \eta \cdot \bw^{OLD}(\Pi_{H' \mapsto G}(e))$ from \Cref{lne:increaseWeights}.
But since the if-statement was true, we have that $\bw^{OLD}(\Pi_{H' \mapsto G}(e)) \leq 2^i \cdot C\APSP$.
We conclude that each edge that is newly embed increases $\|\bw\|_1$ by at most $\eta \cdot 2^i \cdot C\APSP$.

At the beginning of the $i$-th iteration of the for-loop, there are at most $10n/2^i$ edges in in $E(H) \setminus E(H')$.
At the very first iteration $i = 0$, $|E(H)| \le 20 n$ as the max degree of $H$ is at most $10.$
Later, $|E(H) \setminus E(H')| \le 10n / 2^{i-1}$ holds or otherwise the algorithm would terminate after the $(i-1)$-th iteration in \Cref{lne:terminateCase2}.
Thus, during the $i$-th iteration, the foreach-loop in \Cref{lne:foreachUnrouted} iterates over at most $10n / 2^{i-1}$ edges as well.
We can bound the total increase of $\|\bw\|_1$ during the $i$-th iteration by 
\[\frac{10n}{2^{i-1}} \cdot \eta \cdot 2^i \cdot C\APSP = 20n \cdot \eta C\APSP.\]

The total number of iterations is at most $\lfloor \log_2(1/b) \rfloor + 1.$
This establish the second inequality using the definition of $\eta$.

\end{proof}

Note that for every edge $(u,v)$ that is in $E(H) \setminus E(H')$ when the algorithm returns in \Cref{lne:terminateCase2}, the preceding foreach-loop iterated over $(u,v)$ and found that $\textsc{APSP}.\textsc{QueryDist}(u,v) > 2^i \cdot C\APSP$ (as otherwise $(u,v)$ would have been added to $E(H')$). But this implies that 
$\dist_{\bw}(u,v) > 2^i \cdot C = C / b'$
by our choice of $b'$.
To establish correctness, it only remains to use the if-condition preceding \Cref{lne:terminateCase2} and observe that the condition does not hold when $i = 0$.

\paragraph{Correctness (Returning in \Cref{lne:terminateEmbed}).} It is straight-forward to see from \Cref{alg:embedOrSep} that $\Pi_{H' \mapsto G}$ is a correct embedding from $H'$ to $G$ and that $|E(H) \setminus E(H')| \leq 10bn$. It thus only remains to bound the congestion of $\Pi_{H' \mapsto G}$.

\begin{lemma}\label{lma:congestion}
The congestion of $\Pi_{H' \mapsto G}$ is at most $\frac{2\log(2 C \APSP/ b)}{\eta}$.
\end{lemma}
\begin{proof}
Let us fix any edge $e \in E(G)$. Note that each time we add an embedding path in the foreach-loop starting in \Cref{lne:foreachUnrouted} that contains $e$, we increase the weight $\bw_e$ to $(1+\eta)\bw_e$. Since initially, $\bw_e = 1$, we have that after $t$ times that the edge $e$ was used to embed an edge in the foreach-loop, we have that $\bw_e = (1+\eta)^t \geq e^{t\eta/2}$ since $e^x \leq 1+2x$ for $x \in [0,1]$. In particular, if the algorithm embeds $t$ times into $e$ for $t > \frac{2\log(2 C \APSP / b)}{\eta}$, then at the end of the algorithm, we would have $\bw_e > \frac{2 C \APSP}{b}$.

However, note that by the if-condition in \Cref{lne:ifReallyEmbed}, we never embed into an edge $e$ that has weight more than $2^{\log_2(1/b)} \cdot C \APSP = \frac{C \APSP}{b}$ since otherwise the path using this edge has higher weight. We can thus conclude that at the end of the algorithm, $\bw_e \leq (1+\eta)\frac{C \APSP}{b} \leq \frac{2C \APSP}{b}$, which leads to a contradiction. 
\end{proof}

\paragraph{Run  time Analysis.} The for-loop of the algorithm runs at most $O(\log (1/b))$ iterations and in the $i^{th}$ iteration at most $O(n/2^i)$ edges are iterated over in the foreach-loop starting in \Cref{lne:foreachUnrouted}. Thus, the total number of queries to the APSP data structure can be bound by $O(\sum_i n/2^i) = O(n)$. 

The time the algorithm spends updating the weights in \Cref{lne:increaseWeights} can be bound by observing that each edge $e$ has its weight increased only after an additional embedding path was added through $e$; but the congestion is bound by $O(\log n/\eta)$ by \Cref{lma:congestion}, thus the foreach-loop is executed at most $O(n\log n/\eta)$ times over the entire course of the algorithm. This concludes our analysis of the number of updates to the APSP data structure. The runtime analysis of the algorithm follows along the same line of reasoning.

\begin{remark}
\label{rem:extractExpanderSW19}
Our algorithm can be extended to compute expander decompositions, following the approach of \cite{SaranurakW18}.
We refer the reader to this paper for additional background and the necessary definitions.
For readers familiar with \cite{SaranurakW18}, we briefly describe the key step we need to implement:
When  $\textsc{SeparateOrCertify}(G,H,C,b)$ certifies that most edges in the expander $H$ can be embedded into $G$ (and hence by \Cref{lma:folklore_embedding} there are no sparse balanced cuts in $G$) then we need to be able to extract a large expander from $G$ so that we only need to recurse on a small (potentially) non-expanding part
To find an induced subgraph with large expansion, we first produce a new graph $G'$ by adding the edges $E(H)\setminus E(H')$ to $G$.
This ensures that $G'$ is a good expander. We then use the expander pruning of \cite{SaranurakW18} to delete the same edges $E(H)\setminus E(H')$ from $G'$, resulting in a large leftover expander $G''$ with vertex set $V''$. By construction $G[V'']$ is now a large expander.
\end{remark}

\subsection{Extracting the Sparsest Cut}

In order to prove \Cref{theorem:balanceCut}, we now have to show how to extract a sparsest cut from the weight function that is returned in case no embedding is found. We point out that in order to do so it is significantly more convenient to work with an integral weight function $\bw$. We therefore round the weight function that we obtain \Cref{lemma:SeparateOrCertifyBalanced} up which might result in $\|\bw\|_1$ being at most twice as large as stated.

We use the following auxiliary algorithm that finds a cut with few edges crossing given any two vertices at large distance. 

\begin{restatable}{claim}{sepLem}\label{clm:thinLayer}
The procedure $\textsc{FindThinLayer}(G, \bwInt, u, v, D)$ takes a graph $G$ weighted by $\bwInt \in \mathbb{N}_{\geq 1}^{E(G)}$ and two vertices $u,v$ such that $\dist_{\bwInt}(u,v) > D$ for some integer $D > 4 \log_2 \|\bwInt\|_1$. It returns a set of vertices $S \neq \emptyset$ such that $|S| \leq |V|/2$ and $|E_G(S, V \setminus S)| \leq \frac{4  \bwInt(S) \log_2 \|\bwInt\|_1 }{D}$. The algorithm runs in time $O(|E_G(S)| \log |E_G(S)|)$.
\end{restatable}

Given this auxiliary algorithm, we can state the final algorithm and prove our main result, \Cref{theorem:balanceCut}. As described before, we use the algorithm $\textsc{SeparateOrCertify}(G, H, C, \hat{b})$ with a constant degree, constant sparsity expander $H$. It is straight-forward to conclude that $G$ contains no balanced sparse cuts, if the procedure can embed $H$. 

Otherwise, we take the weight function and repeatedly find a separator between the endpoints of edges in $F$ that are far from each other (using the auxiliary algorithm). Note that if there are roughly $b'n$ edges in $F$ at distance roughly $C/b'$, then using the auxiliary algorithm repeatedly with $D \approx C/b'$, produces a cut where the smaller side has $\Omega(|F|) = \Omega(b'n)$ vertices. Using the guarantees from the auxiliary procedure, we further have that the number of edges in the induced cut are at most $\tilde{O}(b'n/ C)$. Thus, the sparsity of the cut must be $\tilde{O}(1/C)$ where $C \approx 1/\psi$ by our choice of parameters.

\begin{algorithm}
$H \gets \textsc{ConstDegExpander}(|V(G)|)$; $C \gets 320\log n/\psi$; \\ 
\If{$\textsc{SeparateOrCertify}(G, H, C, 2b)$ returns $(H', \Pi_{H' \mapsto G})$}{
    \Return $(H', \Pi_{H' \mapsto G})$.\label{lne:retTrivial}
}\Else(\tcp*[h]{i.e. if it returns $(\bwInt, b', F)$}){
    $\wh{\bw} \gets \lceil \bwInt \rceil$. \\
    $X \gets V(G)$.\\
    $D \gets 2C / b'$.\\
    \While{$\exists (u,v) \in H[X] \cap F$ and $|V \setminus X| \le n / 4$}{
        \tcp*[h]{$\dist_{\wh{\bw}}(u, v) > D$} \\
         $S \gets \textsc{FindThinLayer}(G[X], \wh{\bw}, u, v,  D)$.\\
         $X \gets X \setminus S$.
    }
    \Return $V \setminus X$. \label{lne:returnSparseCut}
}
\caption{$\textsc{SparseCutOrCertify}(G, \psi, b)$}\label{alg:mainAlgo}
\end{algorithm}

\APSPreduction*

\begin{proof}%
The case where \Cref{alg:mainAlgo} returns in \Cref{lne:retTrivial} follows directly from \Cref{lemma:SeparateOrCertifyBalanced}, \Cref{thm:constExp} and \Cref{lma:folklore_embedding}.
Let us therefore analyze the remaining case where the algorithm returns in \Cref{lne:returnSparseCut} (the while-loop can be seen to terminate since each iteration shrinks the set $X$ by \Cref{clm:thinLayer} and $X = \emptyset$ trivially has no two vertices at far distance). 

We first prove that the final set $V \setminus X$ has size $b'n \leq |V \setminus X| \leq \frac{3}{4}n$:
\begin{itemize}
    \item $b'n \leq |V \setminus X|$: Initially, $H[X] = H$ and $F \subseteq H$ contains more than $10b' n$ edges by \Cref{lemma:SeparateOrCertifyBalanced}.
    Every edge $(u, v) \in F$ has $\dist_{\wh{\bw}}(u,v) \ge \dist_{\bwInt}(u,v) > C / b'.$ Since the maximum degree of $H$ is $10$, as long as $|V \setminus X| < b'n$, $H[X]$ contains all but $10b'n$ edges from $H.$
    Thus, $H[X] \cap F$ is not empty and the while-loop continues.
    We conclude that $b'n \leq |V \setminus X|$ holds.
    \item $|V \setminus X| \leq \frac{3}{4}n$: Since the while-loop condition allows only invocations of $\textsc{FindThinLayer}$ if $|V \setminus X| \leq n/4$, and since this procedure returns the smaller side of the cut it produces by \Cref{clm:thinLayer} (which is found on $G[X]$), we can conclude that at the end of the algorithm $|V \setminus X| \leq n/4 + n/2 \leq \frac{3}{4}n$.
\end{itemize}
This indicates that $|X| \ge n / 4 \ge b'n / 2 \ge bn$ since $2b \le b' \le \frac{1}{2}.$

Next, we bound the sparsity of the cut $V \setminus X.$
Let $S_1, S_2, \ldots, S_k$ be the sets returned by procedure $\textsc{FindThinLayer}$ one after another over the course of the while-loop, such that $V \setminus X = \cup S_i$. We first observe that these sets are vertex-disjoint since after the $i$-th iteration, the procedure $\textsc{FindThinLayer}$ is invoked on the graph $G_i = G[V \setminus (S_1 \cup \ldots \cup S_i)]$ to find $S_{i+1}$.
Further, the final cut $(X, V \setminus X)$ contains only edges that were previously in a thin layer, i.e. \[
E_G(X, V \setminus X) \subseteq \bigcup_i E_{G_i}(V \setminus (S_1 \cup \ldots \cup S_i), S_i).
\]

It remains to use the guarantee of \Cref{clm:thinLayer} that for each $S_i$, we have $|E_{G_i}(S_i, V \setminus (S_1 \cup \ldots \cup S_i))| \leq \frac{4  \wh{\bw}(S_i) \log_2 \|\bwInt\|_1 }{D}$ and by the vertex-disjointness of $S_1, S_2, \ldots, S_k$, we thus have that
\begin{align*}
    |E_G(X, V \setminus X)| &\leq | \bigcup_i E_{G_i}(S_i, V \setminus (S_1 \cup \ldots \cup S_i))| \leq \sum_i \frac{4  \wh{\bw}(S_i) \log \|\wh{\bw}\|_1 }{D} \\
    &\leq \frac{4  \|\wh{\bw}\|_1 \log \|\wh{\bw}\|_1 }{D} = \frac{8 n \cdot b' \log n}{C}
\end{align*}
where we use $\|\bwInt\| \leq 20n$ from \Cref{theorem:balanceCut} and $\wh{\bw}$ is obtained from rounding up $\bwInt$, and our choice of $D$. Since we have shown that $|X|, |V \setminus X| \geq b'n / 2 \geq bn$, choosing $C = 320\log n/\psi$, we have $\Psi(V \setminus X) = \Psi(X) \leq \psi$, as desired.

We use the disjointness of $S_1, S_2, \ldots, S_k$ to argue that the total time spend in procedure $\textsc{FindThinLayer}$ can be bound by $O(n \log n)$. The remainder of the runtime analysis is trivial given \Cref{lemma:SeparateOrCertifyBalanced}.
\end{proof}

It remains to provide an implementation of $\textsc{FindThinLayer}(G, \bwInt, u, v, D)$ and prove \Cref{clm:thinLayer}.
The algorithm follows a simple ball-growing procedure.
It grows balls from both endpoints $u$ and $v.$
Because the distance between $u$ and $v$ are guaranteed to be large, the procedure takes longer time.
However, these two balls cannot be larger than the entire graph.
There must be a moment that one of the ball grows only by a thin layer.

\sepLem*
\begin{proof}
Since $\dist_{\bwInt}(u,v) > D$ by assumption, we have that at least one of $u$ and $v$ have their ball to radius $D/2$ contain at most half the vertices in $G$. More formally, for some $z \in \{u,v\}$, $|B_{G, \bwInt}(z, D/2)| \leq |V|/2$. We claim that there is a radius $0 < r \leq D/2$, such that taking $S = B(z, r)$ satisfies the above guarantees. 
For this proof, it is convenient to define the following auxiliary function $\Phi(z,r) = \sum_{e \in E} \Phi(z,r,e)$ where the latter functions are defined for all edges $e = (x,y) \in E$ by
\[
   \Phi(z, r, e) = 
    \begin{cases} 
         |\dist_{\bwInt}(z,x) - \dist_{\bwInt}(z,y)|  & \text{if } \dist_{\bwInt}(z,x)  \leq r \text{ and } \dist_{\bwInt}(z,y) \leq r\\
        r - \dist_{\bwInt}(z,x) & \text{if } \dist_{\bwInt}(z,x) \leq r < \dist_{\bwInt}(z,y)\\
        r - \dist_{\bwInt}(z,y) & \text{if } \dist_{\bwInt}(z,y) \leq r < \dist_{\bwInt}(z,x)\\
        0 & \text{otherwise}
    \end{cases}
\]
Here, an edge $e = (x,y) \in E(G)$ contributes the distance between its two endpoints $x$ and $y$ (which is at most $\bw_e$) to $\Phi(z,r,e)$ if both endpoints are fully contained in the ball $B(z,r)$. If neither of the endpoints are contained it contributes $0$. Otherwise, $e = (x,y)$ contributes the distance of the endpoint closer to $z$ to the boundary of the ball. In both cases, $0 \leq \Phi(z, r, e) \leq \bwInt_e$. This means in particular that the weight of edges incident to $B(z,r)$ denoted by $\bwInt(E(B(z,r)))$ is always greater-equal to $\Phi(z,r)$, i.e. $\bwInt(E(B(z,r))) \geq \Phi(z,r)$ for all $r$.

Note further that $\Phi(z, r+1) - \Phi(z,r)$ is exactly $|E_G(B(z, r), V \setminus B(z, r))|$, the number of edges that leave $B(z,r)$. 
To see this, observe that an edge $e = (x, y)$ contributes $1$ to the difference if $\dist_{\bwInt}(x, z) \le r < r+1 \le \dist_{\bwInt}(y, z)$ holds, i.e. $e$ leaves $B(z, r)$.
Otherwise, the contribution of $e$ are identical in both $\Phi(z, r)$ and $\Phi(z, r+1).$
Here we use that $\bwInt$ is integral and so are distances in $G$.

Given this set-up, assume for contradiction that for all $0 < r < D/2$, we have 
\[\Phi(z,r+1) > \left(1+ \frac{4 \log_2 \|\bwInt\|_1}{D}\right) \Phi(z,r).\] 
By induction we have that 
\[\Phi(z, D/2) \geq \left(1+ \frac{4 \log_2 \|\bwInt\|_1}{D}\right)^{D/2 - 1} \Phi(z,1)%
>\|\bwInt\|_1\] where we use that $1+x \geq 2^x$ for $x \in [0,1]$.
This would give a contradiction since $\|\bwInt\|_1 \ge \bwInt(E(B(z,D/2))) \geq \Phi(z,D/2) > \|\bwInt\|_1$.

Therefore, there must be some radius $0 < r < D/2$ such that
\begin{align*}
\Phi(z,r+1) \le \left(1+ \frac{4 \log_2 \|\bwInt\|_1}{D}\right) \Phi(z,r).
\end{align*}
Combining with our previous discussion yields that
\begin{align*}
|E(B(z,r), V \setminus B(z,r))|
&= \Phi(z, r+1) - \Phi(z, r) \\
&\le \frac{4 \log_2 \|\bwInt\|_1}{D}\Phi(z,r) \\
&\le \frac{4 \bwInt(E(B(z,r))) \log_2 \|\bwInt\|_1}{D}.
\end{align*}
We can therefore take $S = B(z,r)$, as desired.

Finally, to compute this cut, we run Dijkstra's algorithm from $u$ and $v$ in parallel and check for the earliest radius $r$ for either of them such that the inequality holds. Thus, the algorithm runs in time $O(|E_G(S)| \log |E_G(S)|)$.
\end{proof}

\printbibliography

\appendix

\section{Reducing Conductance to Sparsity}
\label{sec:condViaSpars}

Here, we prove \Cref{thm:conductanceVersion}. The proof is an adaption of Lemma 5.4 and Theorem 5.5 of \cite{ChuzhoyGLNPS19}.

\paragraph{The Transformation Algorithm.} Our algorithm is essentially a wrapper function around our main result \Cref{thm:fineGrainedReduction}. That is, we first construct a bounded degree graph $\wh{G}$ from $G$, then run the algorithm from \Cref{thm:fineGrainedReduction} on $\wh{G}$. 
If the algorithm certifies that $\wh{G}$ has no balanced sparse cuts, we prove that $G$ has no balanced low-conductance cuts.
Otherwise, if the algorithm returns a sparse cut in $\wh{G},$ we recover a balanced low-conductance cut in $G.$

We first describe the construction of $\wh{G}$ given $G = (V, E)$. Let us assume an arbitrary ordering of the edges incident to each vertex $v \in V$. $\wh{G} = (\wh{V}, \wh{E})$ is constructed as follows:
\begin{enumerate}
    \item For each vertex $v \in V$, create a set of vertices $X_v = \{v_{1}, v_2, \ldots, v_{\deg(v)}\}$, and an $\psi_0$-expander $H_v$ on $X_v$ using \Cref{thm:constExp}.
    Add $H_v$ to $\wh{G}.$
    \item For each edge $e = (u, v) \in E$, we add $(u_i, v_j)$ to $\wh{E}$ if $e$ is the $i^{th}$ (and $j^{th}$) edge incident to $u$ (and $v$, respectively). 
\end{enumerate}
Clearly, $\wh{G}$ has $\vol(G) = 2m$ vertices and each vertex has at most $10$ incident edges, $9$ from the expander and $1$ from the corresponding edge in $G$. 

We now run the algorithm from our main result, \Cref{thm:fineGrainedReduction}, on the graph $\wh{G}$. If the algorithm certifies that no $b$-balanced $\psi$-sparse cut exists in $\wh{G},$ we return the same result for $G$.
Otherwise, we run \Cref{alg:transform} on the returned cut $(A, \overline{A})$ in $\wh{G}$ to obtain a $\Omega(b)$-balanced $\Omega(\psi)$-sparse cut $(S, \overline{S})$ in $G$. It is straight-forward to check that \Cref{alg:transform} is deterministic and runs in time linear in the number of edges of $G$ and thus the runtimes stated in \Cref{thm:fineGrainedReduction} are asymptotically not affected.

\begin{algorithm}
\Return $S = \Set{u \in V}{\Abs{X_u \cap A} \ge \Abs{X_u \setminus A}}$.\label{lne:addS} 
\caption{$\textsc{Transform}(G, \wh{G}, A \subseteq V(\wh{G}))$}\label{alg:transform}
\end{algorithm}

\paragraph{Certifying $G$.} We start by showing that if no $\Omega(b)$-balanced $O(\psi)$-sparse cut is found on $\wh{G}$, then no such cut exists in $G$ either.

\begin{lemma}
\label{lemma:reduceToConstDegNoCut}
Given a balance parameter $b \in (0, 1/4)$, if every cut $(X, \overline{X})$ in $\wh{G}$ with $\Abs{X}, \Abs{\overline{X}} \ge b \cdot |V(\wh{G})|$ has $\Psi_{\wh{G}}(S) \geq \psi$, then every cut $(S, \overline{S})$ in $G$ with $\vol_G(S), \vol_G(\overline{S}) \ge b \cdot \vol(G)$ has  $\Phi_{G}(X) \geq \psi$.
\end{lemma}
\begin{proof}
Let $(S, \overline{S})$ be any cut in $G$ with $\vol_G(S), \vol_G(\overline{S}) \ge b \cdot \vol(G)$. Define $X_S = \cup_{u \in S} X_u$ and $\overline{X_S} = \cup_{u \not\in S} X_u = X_{\overline{S}}.$ Observe that $|E_G(S, \overline{S})| = |E_{\wh{G}}(X_S, \overline{X_S})|$ because the $\psi_0$-expander edges in $\wh{G}$ do not appear in the cut and every cut edge $(u_i, v_j)$ in $\wh{G}$ corresponds to the cut edge $(u, v) \in G.$

By construction of $\wh{G}$, we have that $|X_S| = \vol_G(S)$ and $|\overline{X_S}| = \vol_G(\overline{S})$ and therefore $|X_S|, |\overline{X_S}| \geq b \cdot \vol(G) = b |V(\wh{G})|$ by assumption on $(S, \overline{S})$. Thus $(X_S, \overline{X_S})$ is balanced in $\wh{G}$ and we can use the guarantee that $\Phi_{\wh{G}}(X_S) \geq \psi$. This yields
\begin{align*}
|E_G(S, \overline{S})| = |E_{\wh{G}}(X_S, \overline{X_S})| \ge \psi \cdot \min\{\Abs{X_S}, \Abs{\overline{X_S}}\} = \psi \cdot \min\{\vol(S), \vol(\overline{S})\}.
\end{align*}
\end{proof}

\paragraph{Returning a Sparse Cut.} It remains to prove that the above algorithm transforms any balanced sparse cut in $\wh{G}$ to a balanced low conductance cut in $G.$ We prove this claim in two steps. We first show that the number of edges in the cut $(S, \overline{S})$ in $G$ is comparable to the number of edges in $(A, \overline{A})$ in $\wh{G}$. 

\begin{claim}\label{clm:numCrossingRemainsLow}
$|E_G(S, \overline{S})| = O\left(\Abs{E_{\wh{G}}(A, \overline{A})}\right)$.
\end{claim}
\begin{proof}
Define $X_S = \cup_{u \in S} X_u$. Consider any vertex $u \in V$, we have that the graph $H_u$ contributes at least $\psi_0 \cdot \min\{|X_u \cap A|, |X_u \setminus A|\}$ edges to the cut $\Abs{E_{\wh{G}}(A, \overline{A})}$. But in $\wh{G}$, the number of edges incident to $X_u$ that are in the cut $(S, \overline{S})$ but where previously not in the cut $(A, \overline{A})$ can be at most $\min\{|X_u \cap A|, |X_u \setminus A|\}$ since $H_u$ is contained entirely in $S$ or $\overline{S}$ and only one additional edge is incident to each vertex in $X_u$. 

Thus, we can charge each edge in $E_{H_u}(A, \overline{A})$ with at most $1/\psi_0$ edges from $E_{\wh{G}}(S, \overline{S}) \setminus E_{\wh{G}}(A, \overline{A})$ incident on $u$ and cover all such edges. We conclude that $|E_{\wh{G}}(S, \overline{S})| \le |E_{\wh{G}}(A, \overline{A})| + |E_{\wh{G}}(A, \overline{A})|/\psi_0$, and finally use that $|E_G(S, \overline{S})| = |E_{\wh{G}}(X_S, \overline{X_S})|$ as observed in \Cref{lemma:reduceToConstDegNoCut}.
\end{proof}

Next, we prove that $(S, \overline{S})$ is a balanced cut.

\begin{claim}\label{clm:balancedProd}
If $\Psi_G(A) \leq \psi_0/2$, we have $\vol_G(S) \geq \frac{1}{2}|A|$ and $\vol_G(\overline{S}) \geq \frac{1}{2}|\overline{A}|$.
\end{claim}
\begin{proof}
We prove $\vol_G(S) \geq \frac{1}{2}|A|$ (the proof of $\vol_G(\overline{S}) \geq \frac{1}{2}|\overline{A}|$ is symmetric). Let us assume for the sake of contradiction that $\vol_G(S) < \frac{1}{2}|A|$. We argued before that for every $u \in V$, we have $|E_{H_u}(A, \overline{A})| \geq \psi_0 \cdot \min\{|A \cap X_u|, |A \setminus X_u|\}$.
We again define $X_S = \cup_{u \in S} X_u$ and observe that the fact that $\sum_{u \in S} |A \cap X_u| \le |X_S| = \vol_G(S) < \frac{1}{2}|A|$ implies that $\sum_{u \in S} |A \setminus X_u| \geq |A| - |X_S| > \frac{1}{2}|A|$.
Definition of $S$ also yields that $|A \setminus X_u| \le |A \cap X_u|.$

Combining insights, we conclude 
\[
|E_{\wh{G}}(A, \overline{A})|
\geq \sum_u |E_{H_u}(A, \overline{A})|
\geq \sum_u \psi_0 \cdot \min\{|A \cap X_u|, |A \setminus X_u|\}
\ge  \sum_{u \in S} \psi_0 \cdot |A \setminus X_u|
> \frac{\psi_0}{2}|A|
\]
which implies that $\Psi_G(A) > \psi_0 /2$ which contradicts our assumption, as desired.
\end{proof}

Finally, we combine our insights to prove \Cref{thm:conductanceVersion}.

\begin{proof}[Proof of \Cref{thm:conductanceVersion}.]
We have from the algorithm that $|A|, |\overline{A}| \geq b \cdot |V(\wh{G})| = 2bm$. Therefore, by \Cref{clm:balancedProd}, we produce a cut $(S, \overline{S})$ in $G$ with $\vol_G(S), \vol_G(\overline{S}) \ge b/2 \cdot \vol(G)$. 
By \Cref{clm:numCrossingRemainsLow}, we further have that $|E_G(S, \overline{S})| \leq O(\Abs{E_{\wh{G}}(A, \overline{A})})$ and therefore $\Phi_G(S) = \frac{|E_G(S, \overline{S})|}{\vol_G(S), \vol_G(\overline{S})} = O\left( \frac{\Abs{E_{\wh{G}}(A, \overline{A})}}{\min\{|A|, |\overline{A}|\}}\right) = O(\phi)$ where the last equality stems from the fact that $(A, \overline{A})$ had $\Psi_{\wh{G}}(A) \leq \phi$ by \Cref{thm:fineGrainedReduction}.
\end{proof}

\section{The Constant-Degree Assumption}
\label{sec:constDegAssump}

In this section, we prove that the following assumptions are without loss of generality.

\begin{assumption}\label{assump:constDegree}
When computing a sparse cut with respect to sparsity, we may assume at a cost of a constant factor in the output quality that the input graph $G$ has maximum degree 10.
\end{assumption}

\begin{proof}
Consider obtaining the graph $\wh{G}$ from $G$ by adding $\lceil m/n \rceil$ self-loops to each vertex in $G$. We then invoke \Cref{thm:conductanceVersion} on $\wh{G}$ with $\psi$ and parameter $b$.

Note first that in a connected graph $G$, we have that $\wh{G}\leq 4m$. Further, note that since self-loops do not appear in cuts, we have $E_G(S, \overline{S}) = E_{\wh{G}}(S, \overline{S})$ for all $S$. 

Now, if the algorithm certifies low conductance of $\wh{G}$, we have for each $(S, \overline{S})$ in $G$ where $|S|, |\overline{S}| \geq 4b \cdot n$ that $\vol_{\wh{G}}(S) \geq |S| \lceil m/n \rceil \geq 4b m \geq b \cdot \vol(\wh{G})$. Since $|E_G(S, \overline{S})| = |E_{\wh{G}}(S, \overline{S})| \geq \psi \min\{\vol_{\wh{G}}(S), \vol_{\wh{G}}(\overline{S})\} \geq \frac{1}{3} \psi \min\{|S|, |\overline{S}|\}$. Thus every $4b$-balanced sparse cut has sparsity at least $\frac{1}{3}\psi$.  Otherwise, the algorithm returns a cut $(S, \overline{S})$ of conductance at most $\phi$ in $\wh{G}$. But, we have $\Phi_{\wh{G}}(S) \geq \Psi_{\wh{G}}(S) = \Psi_{G}(S)$ for all $S$.
\end{proof}

\section{A Simple Randomized Algorithm to Construct Low-Degree Expanders}
\label{sec:randExpand}

\begin{algorithm}
Construct an empty graph $H$ on $n$ vertices.\\
\ForEach{$v \in V(H)$}{
    \For{$i = 1, 2, \ldots, k = 80 \log n$}{
     Sample a vertex $u$ from $V(H)$ uniformely and i.i.d. at random.\\ 
     Add edge $(u,v)$ to $H$.
    }
}
\Return $H$
\caption{$\textsc{RandConstDegExpander}(n)$}
\label{alg:randExpa}
\end{algorithm}

\paragraph{The Algorithm.} Here, we provide \Cref{alg:randExpa} which implements the algorithm mentioned in \Cref{rmk:simpleExpander}.

\paragraph{Analysis.} Before we start our analysis, we recall the following Chernoff bound.

\begin{theorem}
Given i.i.d. $\{0,1\}$-random variables $X_1, X_2, \ldots, X_k$, $X = \sum_i X_i$ and any $\delta \geq 0$, we have $P[X \geq (1+\delta) \mathbb{E}[X]] \leq e^{- \frac{\delta^2 \mathbb{E}[X]}{(2+\delta)}}$ and $P[X \leq (1-\delta)\mathbb{E}[X]] \leq e^{- \frac{\delta^2\mathbb{E}[X]}{2}}$.
\end{theorem}

Let us first prove that $H$ has bounded degree.

\begin{claim}
\Cref{alg:randExpa} returns $H$ such that w.h.p., the maximum degree is $O(\log n)$.
\end{claim}
\begin{proof}
Each vertex $u$ is selected as the second endpoint of an edge added to $H$ in the inner for-loop with probability $1/n$ per iteration. As there are $nk$ iterations of this for-loop, and each iteration is independent, we have by the Chernoff bound that each vertex $u$ is at most $k$ times selected with probability at least $1 - e^{- \frac{4 k}{4}} = 1 - e^{-k} = 1- n^{-32}$. 

Since each vertex $u$ has degree equal to $k$ plus the number of times it is sampled, we have that its degree is at most $2k$ with probability at most $1 - n^{-32}$. We obtain our result over all vertices in $H$ by applying a union bound.
\end{proof}

\begin{claim}
\Cref{alg:randExpa} returns a $\Omega(\log n)$-expander $H$ w.h.p.
\end{claim}
\begin{proof}
Consider any set $S$ with $ |S| \leq n/2$. Then, we have that $\mathbb{E}[E_H(S, \overline{S})] \geq \frac{1}{2}|S| k$ since each edge $(u,v)$ sampled when the foreach-loop iterates over a vertex $v \in S$ has $u \not\in S$ with probability at least $\frac{1}{2}$ and there are $|S| k$ such sampling events. Since they are independent, we further have from the Chernoff bound that $P[|E_H(S, \overline{S})| \leq \frac{1}{4}|S|k] \leq e^{-\frac{|S|k}{16}} = n^{-5|S|}$. It is clear that if $|E_H(S, \overline{S})| > \frac{1}{4}|S|k$ then $\Psi_H(S) \geq \frac{k}{4} = \Omega(\log n)$.

The remaining difficulty is that there are an exponential number of cuts so a union bound seems at first hard to apply. However, we observe that there are at most ${\alpha \choose n} \leq \left(\frac{ne}{\alpha}\right)^{\alpha} \leq n^{3\alpha}$ for $\alpha \geq 1$ cuts where the smaller half contains $\alpha$ vertices. As we have proven that a cut is $\psi$-sparse with probability at most $n^{- 5\alpha}$, we can thus conclude by a simple union bound argument that $H$ is not $\Omega(\log n)$-expander with probability at most $\sum_{\alpha \geq 1} {\alpha \choose n} \cdot n^{- 5\alpha} \leq 1/n$.
\end{proof}

\end{document}